\documentclass[12pt,draftcls,onecolumn]{IEEEtran} %TSP single "manuscript.pdf"
\usepackage{dsfont} % for \mathds{N} or \mathds{R}, etc.
\usepackage{graphicx}  %used for figures
\usepackage[cmex10]{amsmath}
\usepackage{psfrag,amssymb,amsfonts}

\newtheorem {definition}{Definition}

\newtheorem {lemma}{Lemma}
\newtheorem {theorem}{Theorem}

\newtheorem {proposition}{Proposition}
\newtheorem {example}{Example}

\date{}

\begin{document}

\title{Random Access Game in Fading Channels with Capture: Equilibria and Braess-like Paradoxes}

\author{Fu-Te Hsu,~\IEEEmembership{Student Member,~IEEE,} and Hsuan-Jung Su*,~\IEEEmembership{Member,~IEEE}
%\thanks{This work was
%supported by the National Science Council, Taiwan, R.O.C., under
%grant NSC XX-XXXX-X-XXX-XXX.}
\thanks{The authors are with the Department of Electrical Engineering
and Graduate Institute of Communication Engineering, National Taiwan University, Taipei, Taiwan, 10617
                  (email: f96942059@ntu.edu.tw, hjsu@cc.ee.ntu.edu.tw)}}

\maketitle

\begin{abstract}
The Nash equilibrium point of the transmission probabilities in a slotted ALOHA system with selfish nodes is analyzed. The system consists of a finite number of heterogeneous nodes, each trying to minimize its average transmission probability (or power investment) selfishly while meeting its average throughput demand over the shared wireless channel to a common base station (BS). We use a game-theoretic approach to analyze the network under two reception models: one is called power capture, the other is called signal to interference plus noise ratio (SINR) capture.
%Contrary to one's intuition,
It is shown that, in some situations, Braess-like paradoxes may occur.
That is, the performance of the system may become worse instead of better when channel state information (CSI) is available at the selfish nodes.
%In particular, a condition for homogeneous nodes under which Braess-like paradoxes occur is analytically presented.
In particular, for homogeneous nodes, we analytically presented that Braess-like paradoxes occur in the power capture model, and
in the SINR capture model with the capture ratio larger than one and the noise to signal ratio sufficiently small.

%In addition, some distributed algorithms converging to the Nash equilibrium in the optimum operating region are presented.
\end{abstract}

\vspace{-4mm}
\begin{center}
   {\underline{\bf \small EDICS}} \hspace{3mm} {\small SPC-BBND, ~SPC-PERF}
\end{center}
\vspace{-6mm}

\section{Introduction}\label{Section_Introduction}

The simplicity of ALOHA \cite{Aloha70} and slotted ALOHA \cite{SlottedAloha72} systems proposed in the 1970s for random access have attracted a large amount of research.
%Classical analyses of slotted ALOHA from the queueing perspective have focused on the stability region. For example, \cite{StabilityMPR88} studied the stability region in channels with multipacket reception (MPR) capability, an infinite number of users, and the single-buffer (also known as without buffer) assumption. The stability region of this problem is now well understood (see \cite{DataNetwork92}). However, the analysis of the stability region of slotted ALOHA with a finite number of users and infinite buffer is quite complicated due to the interaction between multiple queues.
%Tsybakov and Mikhailov \cite{StabilityAloha79} initiated the stability analysis. Sant and Sharma \cite{Sant_Sharma00} analyzed the stability in capture channels which is a special case of the MPR model. Recently Naware, Mergen and Tong \cite{MPR05} considered the MPR model, and derived an inner bound of the stability region. Complete characterization of the closure of stability region is not available even until this time. It is conjectured that the achievable region obtained by backlogged queues coincides with the closure of stability region by possibly empty queues (see \cite{Luo&Ephremides_conjecture_06} and the references therein).
From the system perspective, the earlier works focused on the issues of average throughput and stability of ALOHA systems with homogeneous users. That is, these works usually assumed that all users in the network have the same statistical characteristics, thus only considered the macroscopic (or network-wide average) performance. The readers are referred to \cite{Zorzi_95} and the references therein for the analyses of slotted ALOHA in fading channels with capture.

Recently, MacKenzie and Wicker \cite{Aloha_Selfish01} first considered slotted ALOHA from the user perspective. They assumed that there is no centralized scheduling, and each user acts selfishly to maximize its own utility function. They then analyzed the Nash equilibrium point by game theory \cite{GameTheory94}.
%The selfish behavior is inevitable since
%there is no way to ensure that the same mandated algorithm will be implemented by all users in a distributive environment.
%without a centralized control users may cheat the system in order to improve their own performance, and
Game theory is a useful tool for modeling and studying the interaction of strategic interactions among selfish players.
The most important equilibrium concept in game theory is the Nash equilibrium (named after its inventor John Nash) at which there is no incentive for any player to unilaterally deviate.
%To this end,
The selfish ALOHA system is more robust and scalable for implementation than a system with centralized control.

The analysis of \cite{Aloha_Selfish01} was extended in \cite{Aloha_Selfish08} to the case with heterogenous users whose costs of transmission are not identical.
The behavior of the network throughput at the Nash equilibria as a function of the costs was analyzed. Unfortunately, the cost is not a parameter one can easily control. It is determined by the relative cost of transmission as compared to the value of a success.

Incorporating the availability of CSI into an ALOHA system, \cite{Aloha_JSAC08} investigated the Nash equilibrium points of CSI-dependent transmission probabilities for heterogenous nodes in time-varying channels. In that channel-aware ALOHA system, each node tries to minimize its average transmission probability selfishly while meeting its average throughput demand to the common base station (BS). It was shown that the feasible region of the nodes' throughput demands in the selfish ALOHA is equivalent to the achievable region by a system with centralized control. Moreover, within the feasible region, exactly two Nash equilibrium points exist.
%They also proposed an inner bound on the feasible region and some properties of Nash equilibrium.
%Besides, a distributed algorithm converging to the Nash equilibrium point was provided.
This work considerably extended the analysis of the network model by Jin and Kesidis \cite{Jin&Kesidis02}.

Other related works on CSI-dependent transmission probabilities in channel-aware ALOHA networks include \cite{Aloha_INFOCOM08}\cite{Sabir09}\cite{NE_Threshold_TSP07}. In \cite{Aloha_INFOCOM08}, the analysis of \cite{Aloha_JSAC08} under the collision model was extended to the capture model in static channels. It was shown that while multiple Nash equilibria may exist, one of them is uniformly preferable in the sense of minimum transmission probability.
In \cite{Sabir09}, the authors considered the network model in \cite{Aloha_JSAC08} without the backlogged assumption, and found that, different from \cite{Aloha_JSAC08}, infinitely many equilibrium points may exist with the distributed algorithm if the slotted ALOHA is stable.
A slotted ALOHA system with general multipacket reception (MPR) \cite{StabilityMPR88}\cite{MPR05} was considered in \cite{NE_Threshold_TSP07}. It was shown that, for selfish nodes to maximize their individual utilities (including transmission and waiting costs), the structure of CSI-dependent transmission probabilities is a threshold strategy.

The availability of CSI was used to vary transmission power rather than transmission probability in \cite{Sun&Modiano05}. The authors of \cite{Sun&Modiano05} derived an explicit CSI-dependent power allocation strategy at the Nash equilibrium point in a slotted ALOHA system with capture, under the assumption that channel gain is uniformly distributed. It was shown that as the number of nodes increases, the system performance with the power allocation strategy at the symmetric Nash equilibrium point approaches that with the optimal mandated power allocation strategy.

In \cite{Aloha_SIGMETRICS08}, not only CSI-dependent transmission probabilities but also CSI-dependent transmission power levels were considered.
%in the network model as compared to \cite{Aloha_JSAC08}.
%The authors analyzed the Nash equilibrium points in a quite realistic framework with CSI, power levels, and transmission rate.
Through numerical study, it was shown that the additional allowance of CSI-dependent transmission power levels may make the system performance worse instead of better as compared to that with only the CSI-dependent transmission probabilities \cite{Aloha_JSAC08}.
This kind of counterintuitive phenomena, which demonstrate a performance degradation when more information or resource is added to a noncooperative network, is known as the \emph{Braess paradox}, introduced by Braess in transportation network planning \cite{Braess_ger}.
%The Braess-like paradox does not necessarily happen in a noncooperative network.
%The Braess paradox \cite{Braess_ger} is basically due to the fact that the equilibrium achieved by the selfish strategies of users is not the optimum, despite the addition of roads to the transportation network which should never make the optimal performance (traffic delay) worse. Thus the optimal performance may only be achievable through centralized control.
There are other Braess-like paradoxes discovered in different contexts, for example, in the contexts of queueing network \cite{Cohen_Jeffries97}, computer network \cite{Braessparadox_computernetwork99}\cite{Braessparadox_computernetwork00}\cite{Braessparadox_computernetwork02}\cite{Braessparadox_computernetwork06}, and wireless communication \cite{Altman_Braess08}.

The game-theoretic approach has widely been applied in communication networks.
%recently, and discovered many implications.
For example,
MacKenzie and Wicker \cite{GameTheory_and_the_Design} showed that game theory can be applied to developing self-configuring wireless networks.
Cui, Chen and Low \cite{Distributed_Update_JSAC08} showed a game-theoretic framework for contention-based medium access control.
Lee {\it et al.} \cite{Reverse_engineering_MAC07} revealed the noncooperative nature of random access from MAC reverse-engineering. They discovered that, in the current backoff-based MAC protocol, the users are participating implicitly in a noncooperative game.

In this paper, we also consider the network model in \cite{Aloha_JSAC08}.
Our work analyzes the Nash equilibrium point of transmission probabilities in fading channels under two more general capture models both including the collision model \cite{Aloha_JSAC08} as a special case. We also extend the analysis of \cite{Aloha_INFOCOM08} for static channels to fading channels.
The main contribution of our work that differentiates it from the earlier works studying the CSI-dependent transmission probability in selfish ALOHA \cite{Aloha_JSAC08}\cite{Aloha_INFOCOM08}\cite{NE_Threshold_TSP07} is that we find that Braess-like paradoxes may occur under some situations.
In other words, in some situations, the availability of CSI may degrade the performance (in terms of, e.g., total power consumption, throughput, etc.).
%in a selfish ALOHA system with capture to allow the selfish nodes to employ more flexible strategies.
%Note that in communication systems, the additional availability of CSI usually increases the capacity. At least, it should never make the optimal performance worse. Yet the equilibrium achieved by the \emph{best response} strategies of the selfish nodes is not only not optimal, but also worse (in terms of performance) than that achieved when CSI is not available.
We call this phenomenon a \emph{Braess-like paradox} due to its analogy to the Braess paradox.
To the best of our knowledge, our work is the first to show a Braess-like paradox \emph{analytically} in a random access network, and this paradox does not occur in the collision model considered in \cite{Aloha_JSAC08}.
Such a discovery is important because it was generally believed that the additional availability of CSI should improve, at least not degrade, the network performance.
%This belief is shown to be false recently, and
%Several examples of Braess-like paradoxes in different contexts have also been discovered to avoid wasting resource. In this paper, we present a mathematically rigorous study in the context of random access networks.

% or the ALOHA model considered in \cite{Adireddy_Tong05}.
%For the case with homogeneous nodes, we find that our expression of system throughput
%coincides with that obtained from queueing analysis with our transmission probability replaced by the effective transmission probability in \cite{Dua_TWC08}.
%This result accords with the conjecture in \cite{Luo&Ephremides_conjecture_06} that the achievable region is equivalent to the stability region. %Hence, our work provides a first step toward the understanding that the Braess-like paradoxes may occur even without the backlogged assumption, for which analytical study is quite complicated.

The two capture models we consider in this paper are similar to the capture model in \cite{Gupta_Kumar00} for static channels. Specifically, one model is called the signal to interference plus noise ratio (SINR) capture model in which the BS receives the packet of a node successfully if the node's SINR is larger than the \emph{capture ratio} $b$.
When $b<1$, it is possible for the BS to successfully receive more than one node's packets as in, for example, CDMA systems whose signal quality can be highly increased after despreading. When $b>1$, at most one node's packet is successfully received as in typical narrowband systems which need the SINR to be high enough to operate properly.
The other model is called the power capture model in which the BS receives the packet of a node successfully if the node has the strongest received power which is at least $1+\Delta$ times stronger than the received power of every other node, where $\Delta\ge 0$ models a guard zone to counter interference.

%The feasible regions of users' throughput demands under the two capture models are analyzed, and
%We also characterize the Nash equilibrium
%reveals that there exist at most two Nash equilibrium points for any achievable throughput demands.
Our work also reveals that when CSI is not available to selfish nodes, any achievable throughput demands in the SINR capture model can be achieved by a Nash equilibrium point with its sum of the transmission probabilities of different nodes no larger than a constant which depends only on the capture ratio. As for
%We call the set of these Nash equilibrium points the \emph{optimum operating region}.
the power capture model, our analysis shows that when CSI is not available to selfish nodes and $\Delta=0$ or
when perfect CSI is available to selfish nodes and $\Delta\le\frac{1}{n-1}$, where $n$ is the number of nodes in the network, there exists a unique Nash equilibrium point for any achievable throughput demands.
When $\Delta=\infty$, the number of Nash equilibrium points becomes exactly two.
%The feasible region shrinks as $\Delta$ increases, and we can explicitly characterize the feasible region for the special case when $\Delta=0$, which is known as perfect power capture.

%In particular, SINR capture with the threshold ratio larger than one, the throughput obtained by nodes with perfect CSI is less as compared to the case where nodes have no CSI, under the condition of same transmission probability (equivalently, transmission power consumption).

%\item We provide two fully distributed mechanisms that converge to the Nash equilibrium point.

The remainder of the paper is organized as follows. We formulate our random access game in Section \ref{Section_Problem Formulation}. The analysis of the Nash equilibrium points and the discussion of Braess-like paradoxes are given
under the SINR capture model in Section \ref{Section_SINR_Capture_Model} and under the power capture model in Section \ref{Section_Power_Capture_Model}.
In Section \ref{Section_Distributed_Algorithm}, we provide one distributed mechanism which can make the system converge to the Nash equilibrium.
The paper is then concluded in Section \ref{Section_Conclusion}.

\section{Problem Formulation}\label{Section_Problem Formulation}
We consider a wireless network where $n$ nodes transmit at the same fixed power level $P_T$ to a BS over a shared channel. Time is slotted, and each transmission attempt sends a packet which occupies one time slot. The amount of information contained in a packet is fixed and is the same for all nodes. Thus, for the brevity of analysis, the throughput is defined in terms of the average number of successfully received packets (which is commonly used under the capture or MPR model, e.g., \cite{Aloha_Selfish03}). %\cite{Aloha_Selfish03}\cite{StabilityMPR88}).
All nodes are synchronized so that each transmission attempt starts at the slot boundary. %Whether or not a node will transmit in a time slot is governed by its transmission probability.

At time slot $k$, the signal $y_{k}$ received by the BS is given by
\begin{align}\label{RA_channel}
y_{k}=\sum^n_{i=1}h_{i,k}B_{i,k}d_{i,k}+\eta_k
\end{align}
where
\begin{align*}
B_{i,k}=
\begin{cases}
   1& \text{if node $i$ transmits in time slot $k$}\\
   0& \text{otherwise}
\end{cases},
\end{align*}
$h_{i,k}$ is the channel gain between node $i$ and the BS, $d_{i,k}$ is the signal from node $i$ (with transmission power $P_T$), and $\eta_k$ is the additive noise at the BS.

The channel gains $h_{i,k}$ are assumed to be independent and identically distributed (i.i.d.) among all nodes, fixed within a time slot, and varying from time slot to time slot as in \cite{Aloha_JSAC08}.\footnote{The assumption of independence among the channels of different nodes is justified if the nodes are located far apart, or when the channels have many scatterers surrounding the nodes (e.g., in urban areas).  In addition, given that the nodes communicating with the same BS are usually in similar environments, their channels have similar characteristics. If we further assume that there is open-loop power control to counter the long term average of the channel effects (so the nodes can have fair competition with one another), the assumption of identical channel distributions is also justified. In that case, the transmission power will be different for different nodes. However, this will not affect the basic assumption of the system model, which is for every node to individually maximize its utility, and the essence of the analysis.}
For clarity of the analysis, the channel gain of any particular node is assumed to be an i.i.d. process with respect to time. Note that even if the channel gain of a node is merely assumed to be a stationary and ergodic process, the results of this paper derived based on \emph{stationary strategies} can be shown to hold. The readers are referred to \cite[Section II]{Aloha_JSAC08} for discussion on this assumption.\footnote{The analysis in this paper focuses on the average performance (in a time slot). When CSI is not available, given the assumption of independence among the channels of the nodes, the correlation of individual channels with respect to time does not affect the average performance. When CSI is available, since we assume that the CSI is instantaneous for each time slot, and proportional to the absolute channel gain (as will be discussed later), the inaccuracy of the CSI can be attributed mainly to the quantization error but not the estimation error due to the time variation of the channel. Thus the time correlation again does not affect the average performance.}
%\footnote{For example, the channel gains for each node can be constant over several time slots (called one block), and vary independent from block to block.}
It is further assumed that, at the beginning of time slot $k$, node $i$ may be able to obtain its own instantaneous CSI $z_{i,k}$, which provides an indication of the quality of the current channel between that node and the BS.
For a TDD system, the CSI may be measured by individual nodes based on the signal from the BS (e.g., a periodic pilot signal) and the channel reciprocity.
In an FDD system, the knowledge of CSI may be obtained via feedback from the BS.
Our model only assumes that some sort of instantaneous CSI is available and does not restrict how the CSI is obtained. In practice, instantaneous CSI may be difficult to derive from the feedback from the BS due to the busty nature of random access which may introduce a random time lag between the time the BS measures the CSI and the time of transmission of a node. Thus our system model is more applicable to TDD systems.
Assume the number of possible values of $z_{i,k}$ is $x_i$, and $z_{i,k}$ belongs to the set $\{z_{i1},z_{i2},\dots,z_{ix_i}\}$ with $z_{i1}<z_{i2}<\cdots<z_{ix_i}$. Here the assumption of finite discrete (quantized) CSI is taken for convenience only, and can be relaxed. For the analysis in this paper to hold, the only required property of the mapping from the channel gain to the CSI is that a larger value of CSI corresponds to a higher range of absolute channel gain $|h_{i,k}|$.
%\footnote{For slotted ALOHA systems, due to random transmissions of the nodes, the set of users simultaneously transmitting in one time slot is different from that in the next time slot. Thus, in every time slot, the received SINRs of the nodes remain unknown until the transmissions actually happen. It is therefore impractical to assume that the CSI corresponds to the received SINR.}
By excluding the CSI values observed with probability zero, we can assume that the probability of observing each element in the set $\{z_{i1},z_{i2},\dots,z_{ix_i}\}$ is larger than zero.
Note that the case where CSI is not available can be seen as a special case with only one possible CSI value.
%Here, finite CSI signals are assumed (i.e. partial CSI) for convenience only.

%The results in this work can also apply to independent block fading channels since our focus is on Nash equilibrium in stationary transmission.
%Excluding CSI signals with zero probability, we can assume the probability of observing a particular CSI signal in $\{z_{i1},z_{i2},\dots,z_{ix_i}\}$ is larger than zero
\subsection{Two Capture Models}
\begin{enumerate}
\item The SINR capture model:
in this model, node $i$'s packet will be successfully received at time slot $k$ if
$SINR_{i,k}>b$, where $b$ is the capture ratio, and $SINR_{i,k}$ is the SINR of node $i$ at time slot $k$ given by
\begin{align}
SINR_{i,k}=\frac{B_{i,k}|h_{i,k}|^2P_T}{\sum_{j\neq{}i}B_{j,k}|h_{j,k}|^2P_T+N_0}
\end{align}
with
%$B_{i,k}$ is a binary value, which equal to $1$ if node $i$ transmits at time slot $k$, otherwise equal to $0$, and
$N_0$ being the power of the additive noise at the BS. If we set $N_0=0$ and the capture ratio $b=\infty$, this capture model becomes the collision model.
\item The power capture model:
in this model, node $i$'s packet is successfully received at time slot $k$ if %its received power $P_i$ satisfies
\begin{align}\label{eq_power_capture_model}
%P_i > (1+\Delta)\cdot \max_{j\neq i}\{P_j\},
B_{i,k}|h_{i,k}|^2> \max_{j\neq i}\{(1+\Delta)B_{j,k}|h_{j,k}|^2\},
\end{align}
where $\Delta\ge 0$. When $\Delta=\infty$, we have the collision model; and when $\Delta=0$, we have the perfect power capture model for which the packet with the highest received power is always captured.
\end{enumerate}

Throughout this paper, we will only consider that the channels are i.i.d. \emph{Rayleigh fading} \cite{Proakis_COM}. This is the most commonly used model in wireless communications in urban areas. Hence, $|h_{i,k}|^2$ are i.i.d. \emph{exponential} random variables.

\subsection{Noncooperative Game Formulation}
In the network, each node tries to minimize its average transmission probability (or equivalently, average power investment) and selfishly makes the decision whether to transmit or not according to the current CSI, while meeting the \emph{average throughput demand} (in packets per slot), denoted $\rho_i$ for node $i$.
%The throughput demand of each node may be dictated by its application such as voice or video, or packet arrival rate in a sensor network.
It is further assumed that all nodes always have packets buffered for transmission at any time.

This system can be modeled as a noncooperative game with constraints which are the average throughput demands. The selfish nodes are the players, and the action of a player (node) at every time slot is to transmit or not. For generality, in order to meet any average throughput demand while minimizing the average transmission probability, the decision whether to transmit or not is relaxed from being deterministic to being probabilistic. To this end, an action is defined as transmission with a certain probability. With the i.i.d. channel gain processes, we focus on \emph{stationary} transmission strategies (as in \cite{Aloha_JSAC08}) which depend on the \emph{current CSI}. Thus, we let $\textbf{s}_i=(s_{i1},s_{i2},\ldots,s_{ix_i})\in{}[0,1]^{x_i}$ denote node $i$'s transmission strategy such that it transmits with probability $s_{im}$ (the $m$-th entry of $\textbf{s}_i$) when the CSI is $z_{im}$. $\{ s_{i1},s_{i2},\ldots,s_{ix_i} \}$ also defines the action space of node $i$.
Besides the actions of transmission with certain probabilities, the other action of node $i$ is to adjust its transmission strategy $\textbf{s}_i$ such that it can sustain the average throughput demand while minimizing the average transmission probability denoted by $p_i$. The Nash equilibria of the stationary transmission strategies and how the nodes can adjust their transmission strategies in this noncooperative game to arrive at an equilibrium will be discussed later.
%Finally, the utility of node $i$ should be a decreasing function of the average transmission probability (power investment) that can sustain the average throughput demand.

Due to the constraints (average throughput demands), the average transmission probabilities (power investments) of the nodes are nonzero (except for the trivial case with zero throughput demand), and the interaction between nodes is through their mutual interference or competition to have the highest received power. We now discuss the best response strategies for each selfish node in the network.

\begin{definition}
A stationary \emph{threshold strategy} for node $i$ has the form $\textbf{s}_i=(0,\ldots,0,s_{im},1\ldots,1)$. That is, the node always transmits when the CSI is larger than the threshold $z_{im}$, and never transmits when the CSI is smaller than $z_{im}$, while the transmission probability when the CSI is $z_{im}$ is $s_{im}$.
\end{definition}

For example, $\textbf{s}_i=(0,\ldots,0,0.5,1)$ means that node $i$ always transmits when the CSI is the largest one, and transmits with probability $0.5$ when the CSI is the second largest one. For the other CSI values, node $i$ does not transmit.

%It can be proved that the stationary threshold strategy is the best response transmission strategy, under both the power capture and the SINR capture models, for each selfish node to minimize its average transmission probability while meeting the average throughput demand. This is because transmission at higher CSI can get packet success with higher probability.
%To be more specific,
%since the channels of different nodes are independent, for a particular node which does not know the CSI of the other nodes, no matter at what time slot this node transmits and what the other nodes' transmission strategies are, the packet success probability of this node will be affected by the other nodes through the \emph{average interferences} they cause. Thus, for the node in consideration, transmitting when its CSI is higher will result in higher average SINR and hence higher success probability in the SINR capture model. With a similar argument, transmitting at higher CSI will also result in higher success probability in the power capture model.

%In the following, we prove mathematically that the stationary threshold strategy is the best response transmission strategy.
%This proof by contradiction is similar to the proof for the collision channels in \cite{Aloha_JSAC08}.
We have the following proposition as in \cite{Aloha_JSAC08} that the stationary threshold strategy is the best response transmission strategy for each node. The reason is that transmitting at higher CSI will result in higher probability of packet success (or higher throughput), hence more power saving.
To be more specific, since the channels of different nodes are independent, for a particular node which does not know the CSI of the other nodes, no matter at what time slot this node transmits and what the other nodes¡¦ transmission strategies are, the packet success probability of this node will be affected by the other nodes through the average interferences they cause. Thus, for the node in consideration,
transmitting when its CSI is higher will result in higher average SINR and hence higher success probability in the SINR capture model. With a similar argument, transmitting at higher CSI will also result in higher success probability in the power capture model.
Since the proof is similar to that of \cite[\emph{Lemma 1}]{Aloha_JSAC08}, we omit it.
\begin{proposition}\label{best-response}
The \emph{best response} transmission strategy for each selfish node in terms of minimizing its average transmission probability (or average power investment) while meeting the average throughput demand, is a threshold strategy under the power capture and the SINR capture models.
\end{proposition}

\emph{Remarks}:
\begin{enumerate}
\item
%It can be easily verified that a threshold strategy $\textbf{s}_i$ uniquely determines the average transmission probability $p_i$ of node $i$ (or see \cite{Aloha_JSAC08}), and the inverse is true with the selection of the CSI set such that all CSI values in the set have nonzero probabilities.
As a result of \emph{Proposition \ref{best-response}}, we have that a threshold strategy $\textbf{s}_i=(s_{i1},s_{i2},\ldots,s_{ix_i})=(0,\ldots,0,s_{im},1\ldots,1)$ uniquely determines the average transmission probability $p_i$ of node $i$ by
\begin{align}
p_i=s_{im}P_i(m)+\sum^{x_i}_{j=m+1} P_i(j)
\end{align}
where $P_i(j)$ is the probability of occurrence of $z_{ij}$ for node $i$.
Therefore, we will analyze the Nash equilibrium in terms of $p_i$ (as in \cite{Aloha_JSAC08}) in the remainder of this paper.

\item
The best response strategy for a node to adjust its average transmission probability (hence the corresponding stationary threshold strategy) in reaction to the given strategies of the other nodes, is to equalize the throughput achieved by the average transmission probability with the average throughput demand (so the average transmission probability is minimized) \cite{Aloha_JSAC08}.
%In this paper we will focus on analyzing the Nash equilibria achievable by the best response strategies of the nodes in this noncooperative game. The readers are referred to \cite{Sabir09}, \cite{Aloha_JSAC08} and \cite{Distributed_Update_JSAC08} for the distributed algorithms (which can be straightforwardly modified to adjust the average transmission probabilities in our case) and their best response dynamics, in random access games converging to the (better) Nash equilibria.
\item
In the case when CSI is not available (equivalently, there is only one possible CSI value), the threshold strategy of node $i$ becomes random transmission with probability $p_i$ at every time slot. The action space in this case (with only one possible action) is apparently smaller than that of the case when CSI is available. In the limiting case when perfect CSI is available (e.g., the CSI takes the exact value of $|h_{i,k}|$ for node $i$ at time slot $k$, that is, there are an infinite number of possible CSI values),
the threshold strategy of node $i$ becomes transmission only if the CSI is above the threshold which is selected such that the average transmission probability is $p_i$.
In the sequel, only the limiting cases without CSI and with perfect CSI are considered for the brevity of analytically studying the Nash equilibrium points and presenting the Braess-like paradoxes.
%(which occurs for general CSI availability).

%The threshold in the threshold strategy can be uniquely determined once the average transmission probability is determined.
%In practice, the threshold and the transmission probability can be obtained together by the distributed mechanisms discussed in \rsection{Distributed Algorithm}.
\end{enumerate}

\subsection{Nash Equilibria}
Let $\textbf{p}_{-i}$ represent the vector of the transmission probabilities of all nodes except node $i$, and $r_i(p_i,\textbf{p}_{-i})$ represent the average throughput of node $i$ when it transmits with probability $p_i$ given that the other nodes transmit with probability vector $\textbf{p}_{-i}$.
%Our problem can be stated as a noncooperative game as follows: the selfish nodes are the players; the transmission probability $p_i\in[0,1]$ determines the action for (player) node $i$;
For the noncooperative game in consideration, we define the transmission probability vector $\textbf{p}=(p_1,\ldots,p_n)\in[0,1]^n$ as an action profile. The utility function for node $i$, given that the other nodes transmit with probability vector $\textbf{p}_{-i}$, is defined as $U_i(p_i,\textbf{p}_{-i})=1-p_i$ (which may be seen as the power left for node $i$).

We give the definition of the (constrained) Nash equilibrium point in our noncooperative game.
\begin{definition}\label{def_NE}
An action profile $\textbf{p}$ is a (constrained) Nash equilibrium point if for all $i=1,\ldots,n$, we have
\begin{align}
\begin{cases}
r_i(p_i,\textbf{p}_{-i})\geq\rho_i\\
U_i(p_i,\textbf{p}_{-i})\ge U_i(\tilde{p}_i,\textbf{p}_{-i}),~\forall\tilde{p}_i\in\{\tilde{p}_i:r_i(\tilde{p}_i,\textbf{p}_{-i})\geq\rho_i\},
\end{cases}
\end{align}
where $\rho_i$, the average throughput demand, defines a constraint. %$U_i(p_i,\textbf{p}_{-i}) = U_i(\textbf{p})$, and $U_i(\tilde{p}_i,\textbf{p}_{-i})$ is similarly defined by replacing $p_i$ in $\textbf{p}$ with $\tilde{p}_i$.

Equivalently, $\textbf{p}$ is a Nash equilibrium point if
\begin{align}
p_i\in\arg\min_{0\le\tilde{p}_i\le 1}\{\tilde{p}_i:r_i(\tilde{p}_i,\textbf{p}_{-i})\geq\rho_i\},~\forall\ i=1,\ldots,n.
\end{align}
\end{definition}
The above expression means that at a Nash equilibrium point \textbf{p}, each node $i$ would not prefer to deviate from its choice of transmission probability.
It should be noted that our problem is a noncooperative game with constraints,
so there are additional constraints in defining our Nash equilibrium point that differs from the conventional Nash equilibrium point.\footnote{The throughput constraint, which is a form of quality of service (QoS) guarantee, can be incorporated into the utility function by a step function similar to the utility function representation of the QoS in \cite{Goodman00}. For example, we can let the utility function be $U_i(p_i,\textbf{p}_{-i})=Q_i(r_i)\cdot(1-p_i)$, where $Q_i(r_i)=1$ if $r_i\ge \rho_i$, and $Q_i(r_i)=0$ otherwise. With this utility function, a node which can not meet its throughput demand has utility 0. This unconstrained model is more general because it can handle the situation where the system can not sustain all nodes' throughput demands and some nodes will have zero utility. On the other hand, its Nash equilibria are much more difficult to analyze because for a node that can not meet its throughput demand, taking any transmission probability will result in zero utility, but different transmission probabilities will have different impacts on the other nodes' throughputs and utilities. The constrained model focuses on the case where all nodes' throughput demands can be met, and is more straightforward to analyze.}
%The setting of our (constrained) Nash equilibrium point is more appropriate to guarantee the throughput and then analysis the achievable throughput demands.}
%Due to the constraints (average throughput demands), the utility is upper bounded (i.e., the average power investment is bounded away from zero), and the interaction between nodes is through their mutual interference or competition to have the highest received power.

Since $r_i(p_i,\textbf{p}_{-i})$ and the utility function are nondecreasing functions of $p_i$ when $\textbf{p}_{-i}$ is given under both the power capture and the SINR capture models\footnote{This is intuitive, and can be verified by the analytical expressions of $r_i(p_i,\textbf{p}_{-i})$ in (\ref{thrpt_eq_no_CSI}) and (\ref{eq_NE_rhoi}) for SINR and power capture cases, respectively, without CSI; and in (\ref{thrpt_eq_CSI_SINR_capture}) and (\ref{thrpt_eq_CSI_power_capture}) (or (\ref{eq_rho'_homogeneous}), for homogeneous nodes) for the cases with perfect CSI.}, it follows that
an average transmission probability vector $(p_1,\ldots,p_n)$ (where $p_i \in [0,1], \forall i$) is a Nash equilibrium point for the average throughput demands $(\rho_i,\ldots,\rho_n)$ if and only if it is a solution to the set of equations
\begin{align}
r_i(p_i,\textbf{p}_{-i})=\rho_i, \; i=1,\ldots,n.
\end{align}

%Two type of CSI signals are used in this paper, our focus is on nodes with no CSI signal, and we use nodes with perfect CSI signals to point out the potential of Braess-like paradox.

%In the following, we will analyze the Nash equilibrium point for the power capture model and the SINR capture model.

\section{SINR Capture Model}\label{Section_SINR_Capture_Model}
\subsection{Equilibrium point analysis for SINR capture}\label{Subsection_Equilibrium_Point_SINR}
We first consider the case without CSI. In a given time slot, assuming that $n$ nodes simultaneously transmit to the BS, the probability that the packet from a particular node (say, node 1) is successfully received is given by \cite{Dua_TWC08}
\begin{align}
Pr\left[|h_1|^2>b\sum^n_{i=2}|h_i|^2+b\frac{N_0}{P_T}\right]
%&=Pr\left[\sum^n_{i=2}|h_i|^2-\frac{|h_1|^2}{b}<-\frac{N_0}{P_T}\right] \notag \\
&=\left(\frac{1}{1+b}\right)^{n-1}e^{-b\frac{N_0}{P_T}}.
\end{align}
%where the time index of the channel gains is omitted for brevity and the last equality is obtained as follows.
%Let the probability density function (PDF) of $|h_1|^2$ be $f_{|h_1|^2}(x)=e^{-x}, x\ge 0$, then $f_{-\frac{|h_1|^2}{b}}(x)=be^{bx}, x\le 0$, and the PDF of the sum of i.i.d. $|h_i|^2$ is given by $f_{\sum^n_{i=2}|h_i|^2}(x)=\frac{x^{n-2}}{(n-2)!}e^{-x}, x\ge 0$. It follows that the PDF of
%$\sum^n_{i=2}|h_i|^2-\frac{|h_1|^2}{b}$ for $x\le 0$ is given by
%\begin{align*}
%f_{\sum^n_{i=2}|h_i|^2-\frac{|h_1|^2}{b}}(x)&=\int^{\infty}_0e^{-t}\frac{t^{n-2}}{(n-2)!}be^{b(x-t)}dt\notag\\
%&=\left(\frac{1}{1+b}\right)^{n-1}be^{bx}~(\text{for}~x\le 0).
%\end{align*}
%Therefore,
%\begin{align*}
%Pr\left[\sum^n_{i=2}|h_i|^2-\frac{|h_1|^2}{b}<-\frac{N_0}{P_T}\right]
%&=\int^{-\frac{N_0}{P_T}}_{-\infty}f_{\sum^n_{i=2}|h_i|^2-\frac{|h_1|^2}{b}}(x) dx\\
%&=\left(\frac{1}{1+b}\right)^{n-1}e^{-b\frac{N_0}{P_T}}.
%\end{align*}

We first give the following lemma about the average throughput.
\begin{lemma}\label{Through_lemma}
Under the SINR capture model with capture ratio $b$, and i.i.d. Rayleigh fading channels between all nodes and the BS, we have the average throughput of node $i$ when the transmission probability vector $\textbf{p}=(p_1,\ldots,p_n)$ and no CSI is available to all nodes:
\begin{align}\label{thrpt_eq_no_CSI}
r_i(p_i,\textbf{p}_{-i})=e^{-b\frac{N_0}{P_T}}p_i\prod_{j\neq{}i}\left(1-\frac{bp_j}{1+b}\right).
\end{align}
\end{lemma}
\begin{proof}
Let $(x_1,\ldots,x_k)\in{}I_{-\{y_1,\ldots,y_s\}}$ denote $x_1<\cdots<x_k$, all belonging to the node index set $I_{-\{y_1,\ldots,y_s\}}\triangleq\{1,2,\ldots,n\}\setminus \{y_1,\ldots,y_s\}$, where $\setminus$ denotes the set minus operator.
The throughput of node $i$ can be computed as
\begin{align*}
r_i(p_i,\textbf{p}_{-i})=&p_i\cdot\prod_{j\in{}I_{-i}}(1-p_j)\cdot e^{-b\frac{N_0}{P_T}}\\
&+p_i\cdot\sum_{j\in{}I_{-i}}\left(p_j\prod_{k\in{}I_{-\{i,j\}}}(1-p_k)\right)\cdot\left(\frac{1}{1+b}\right)e^{-b\frac{N_0}{P_T}}\\
&+p_i\cdot\sum_{(j,k)\in{}I_{-i}}\left(p_jp_k\prod_{l\in{}I_{-\{i,j,k\}}}(1-p_l)\right)\cdot\left(\frac{1}{1+b}\right)^2e^{-b\frac{N_0}{P_T}}\\
&+\cdots\\
&+p_i\left(\prod_{j\in{}I_{-i}}p_j\right)\cdot\left(\frac{1}{1+b}\right)^{n-1}e^{-b\frac{N_0}{P_T}}\\
%=&e^{-b\frac{N_0}{P_T}}p_i\left\{\prod_{j\in{}I_{-i}}(1-p_j)\right.\\
%&+\sum_{j\in{}I_{-i}}\left(\left(\frac{p_j}{1+b}\right)\prod_{k\in{}I_{-\{i,j\}}}(1-p_k)\right)\\
%&+\sum_{(j,k)\in{}I_{-i}}\left(\left(\frac{p_j}{1+b}\right)\left(\frac{p_k}{1+b}\right)\prod_{l\in{}I_{-\{i,j,k\}}}(1-p_l)\right)\\
%&+\cdots\\
%&\left.+\prod_{j\in{}I_{-i}}\left(\frac{p_j}{1+b}\right)\right\}\\
=&e^{-b\frac{N_0}{P_T}}p_i\prod_{j\neq{}i}\left[\left(\frac{p_j}{1+b}\right)+\left(1-p_j\right)\right]\\
=&e^{-b\frac{N_0}{P_T}}p_i\prod_{j\neq{}i}\left(1-\frac{bp_j}{1+b}\right).
\end{align*}
\end{proof}

%In fact, we have a more general result when the node $i$ transmits with fixed power $P^{(i)}_T$ (or the power loss due to the distance between the node $i$ and the control center is taken into account) as follows, whose proof is similar, so we omit the proof here.
%\begin{theorem}
%Under SINR capture model with capture ratio $b$, and Rayleigh fading channels between nodes and the control center (satisfying the \emph{Assumption 1}), we have the throughput of node $i$ when no CSI signals available for all nodes if the node $i$ transmits with fixed power $P^{(i)}_T$:
%\begin{align}
%r_i(p_i,\textbf{p}_{-i})=Rp_i\prod_{j\neq{}i}\left(1-\frac{bp_j}{\frac{P^{(i)}_T}{P^{(j)}_T}+b}\right)
%\end{align}
%\end{theorem}

\begin{definition}\label{def_achievable}
An average throughput demand vector $(\rho_1,\ldots,\rho_n)$ is called \emph{achievable} if there is a Nash equilibrium point for it (i.e., satisfying \emph{Definition \ref{def_NE}}). The set of all achievable average throughput demand vectors is called the \emph{feasible} throughput region.
\end{definition}

By \emph{Lemma \ref{Through_lemma}}, \emph{Definition \ref{def_NE}} and \emph{Definition \ref{def_achievable}}, an average throughput demand vector $(\rho_1,\ldots,\rho_n)$ is achievable under the SINR capture model if there exists a Nash equilibrium point $(p_1,\ldots,p_n)$ such that $\rho_i=e^{-b\frac{N_0}{P_T}}p_i\prod_{j\neq i}(1-\frac{bp_j}{1+b}), \forall i$. We observe that this expression is similar to that in \cite{Aloha_JSAC08} for the collision model, thus the following result can be easily obtained with a proof similar to that in \cite[\emph{Theorem 3}]{Aloha_JSAC08}. The proof is omitted for conciseness.
%The proof is given in the appendix for completeness.

\begin{theorem}\label{Optimum operating region}
For the SINR capture model with capture ratio $b$, there are at most two Nash equilibrium points for any achievable throughput demands $(\rho_1,\ldots,\rho_n)$ when no CSI is available to all nodes, and exactly one of the Nash equilibrium point can be achieved with
\begin{align*}
\sum^n_{i=1}p_i\leq{}\frac{b+1}{b}.
\end{align*}
\end{theorem}

In the case when perfect CSI is available, i.e., the CSI takes the exact value of absolute channel gain, the best response strategy (threshold strategy) for node $i$ is to transmit only when its CSI is larger than a threshold $T_i$. Assume that we have the Nash equilibrium point $(p_1,\ldots,p_n)$ for the throughput demands $(\rho_1,\ldots,\rho_n)$. Then $T_i$ must satisfy $\int^{\infty}_{T_i}e^{-x_i}dx_i=e^{-T_i}=p_i$.
When there are $s$ nodes in the network and all of them have perfect CSI, the probability that these $s$ nodes simultaneously transmit to the BS, and the packet from a particular node (say, node $i$) gets successfully received is
\begin{equation}\label{success_prob_csi}
\int^{\infty}_{T_s}\!\cdots\!\int^{\infty}_{T_1}\!\left(\!\int^{\infty}_{\max\left\{T_i,~b\left(\sum^{s}_{j=1 \atop j\neq i}\!x_j\!+\!\frac{N_0}{P_T}\right)\right\}}\!e^{-x_i}dx_i\!\right)\!e^{-x_1}dx_1\!\cdots{}\!e^{-x_{s}}dx_{s}.
\end{equation}
This expression is very complicated due to the $\max\{\cdot,\cdot\}$ that accounts for the situation where the threshold $T_i$ is already high enough to guarantee successful reception of node $i$'s packet.
However, if % $b$ is large enough such that $b\left(\sum^{s}_{j=1 \atop j\neq i}\!T_j\!+\!\frac{N_0}{P_T}\right)\geq T_i$,
we have
\begin{align}\label{T_i_Assumptions}
b\left(\sum^{s}_{j=1 \atop j\neq i}\!T_j\!+\!\frac{N_0}{P_T}\right)\geq T_i,
\end{align}
(\ref{success_prob_csi}) can be simplified to
\begin{equation}\label{simplified_success_prob_csi}
\prod^{s}_{j=1\atop j\neq i}\left(\frac{p_j^{b+1}}{b+1}\right) e^{-b\frac{N_0}{P_T}}.
\end{equation}
Note that in practical systems (e.g., narrowband systems for which $b>1$), (\ref{T_i_Assumptions}) is usually true when $s \geq 2$.
(\ref{T_i_Assumptions}) is satisfied except there is one particular $T_i$ which is sufficiently larger compared to the other $T_j$'s. Because the transmission probability $p_i$ is exponentially decreasing in $T_i$, this implies that $p_i$ is sufficiently small compared to the other $p_j$'s, and the throughput demand for node $i$ could be nearly zero. In that case, node $i$ could be removed from the analysis with little effect.
Therefore, the extra conditions (\ref{T_i_Assumptions}) on $T_i$'s are usually satisfied.
%for most scenarios. For a few cases that do not meet the conditions require being studied case by case.
When $s=1$, $b\left(\frac{N_0}{P_T}\right)\geq T_i$ may not be true in many situations. Thus, using (\ref{simplified_success_prob_csi}) only for $s \geq 2$, the throughput of node $i$ at the Nash equilibrium is equal to the demand
\begin{align}
r_i(p_i,\textbf{p}_{-i})=&\rho_i \notag\\
=&\prod_{j\in{}I_{-i}}(1-p_j)\cdot \int^{\infty}_{\max\left\{T_i,~b\left(\frac{N_0}{P_T}\right)\right\}}\!e^{-x_i}dx_i \notag\\
&+\sum_{j\in{}I_{-i}}\left(p_j^{b+1}\prod_{l\in{}I_{-\{i,j\}}}(1-p_l)\right)\cdot\left(\frac{1}{b+1}\right)e^{-b\frac{N_0}{P_T}}\notag\\
&+\sum_{(j,l)\in{}I_{-i}}\left(\left(p_jp_l\right)^{b+1}\prod_{l\in{}I_{-\{i,j,l\}}}(1-p_l)\right)\cdot\left(\frac{1}{b+1}\right)^2e^{-b\frac{N_0}{P_T}}\notag\\
&+\cdots \notag \\
&+\left(\prod_{j\in{}I_{-i}}p_j\right)^{b+1}\cdot\left(\frac{1}{b+1}\right)^{n-1}e^{-b\frac{N_0}{P_T}}\notag\\
=&e^{-b\frac{N_0}{P_T}}\prod_{j\neq{}i}\left(\frac{p^{b+1}_j}{b+1}+(1-p_j)\right)+\prod_{j\in{}I_{-i}}(1-p_j)\cdot \min \left\{p_i-e^{-b\frac{N_0}{P_T}}, ~0\right\},\label{thrpt_eq_CSI_SINR_capture}
\end{align}
where the last equality is obtained using an approach similar to the proof of \emph{Lemma \ref{Through_lemma}}.
%Note that this throughput expression is independent of the average transmission probability $p_i$ of node $i$ due to the assumption $b\left(\sum^{k}_{j=1 \atop j\neq i}\!T_j\!+\!\frac{N_0}{P_T}\right)\geq T_i$. When this condition is not valid (i.e., when $T_i$ is higher and $p_i$ is smaller), the throughput derived using the correct packet success probability (\ref{success_prob_csi}) will be increasing in $p_i$, when $\textbf{p}_{-i}$ is given. Overall, the throughput is a nondecreasing function of the average transmission probability.

To analyze the Nash equilibrium when perfect CSI is available is quite difficult due to the complicated equation (\ref{thrpt_eq_CSI_SINR_capture}), so we only show the existence of Nash equilibra for the case with homogeneous nodes. (That is, the throughput demands are $(\rho_1,\ldots,\rho_n)=(\rho,\ldots,\rho)$, the Nash equilibrium point is $(p_1,\ldots,p_n)=(p,\ldots,p)$, and the threshold is $T$ for all nodes such that $\int^{\infty}_{T}e^{-x_i}dx_i=e^{-T}=p, ~\forall i$.)
In this case, when $b\geq 1$, the throughput of a particular node at the Nash equilibrium point can be computed by (\ref{thrpt_eq_CSI_SINR_capture}) as
\begin{equation}\label{thrpt_eq_CSI_homo}
r_i(p, \ldots, p)=\rho=\left[(1-p)+\frac{p^{b+1}}{b+1}\right]^{n-1}e^{-b\frac{N_0}{P_T}}+(1-p)^{n-1}\cdot \min \left\{p-e^{-b\frac{N_0}{P_T}}, ~0\right\}.
\end{equation}

Because $\frac{dr_i}{dp}<0$ for $p>e^{-b\frac{N_0}{P_T}}$ and $\frac{dr_i}{dp}>0$ when $p$ is sufficiently small (i.e., $\lim_{p\rightarrow 0^+}\frac{dr_i}{dp}>0$), there is a particular $p^*\in\left[0,e^{-b\frac{N_0}{P_T}}\right]$ achieving the maximum of (\ref{thrpt_eq_CSI_homo}), denoted by $\rho_{max}$. Therefore, when perfect CSI is available to homogeneous nodes and given the throughput demand $\rho\le\rho_{max}$, there is at least one Nash equilibrium. Through simulations, we found that there are at most two Nash equilibria even in the case with heterogeneous nodes. However, we are not able to provide a rigorous proof.

\subsection{Braess-like paradox for SINR capture}\label{Subsection_Paradox_SINRcapture}
In this subsection, we will present analytically a Braess-like paradox for the case with homogeneous nodes.
We have the following theorem for the interference limited situation, i.e., $\frac{N_0}{P_T}$ is sufficiently small:

\begin{theorem}\label{Braess_SINR}
With the same average transmission probability $p$ (or average power investment), the throughput of homogeneous nodes with perfect CSI is not larger than that of homogeneous nodes with no CSI when $1 < b < \infty$ and $\frac{N_0}{P_T}$ is sufficiently small.
%$b \nrightarrow \infty$,
\end{theorem}
\begin{proof}
%$b \nrightarrow \infty$ and
Since $\frac{N_0}{P_T} \rightarrow 0$, we have $e^{-b\frac{N_0}{P_T}} \rightarrow 1$. By \emph{Lemma \ref{Through_lemma}} and (\ref{thrpt_eq_CSI_homo}), we need to show that for $b>1$, we have
\begin{align*}
&p\left(1-\frac{b}{b+1}p\right)^{n-1} \ge \left((1-p)+\frac{p^{b+1}}{b+1}\right)^{n-1}-(1-p)^{n}\\
\Leftrightarrow&p\left(1-\frac{b}{b+1}p\right)^{n-1}+(1-p)^{n} \ge \left((1-p)+\frac{p^{b+1}}{b+1}\right)^{n-1}.
\end{align*}
Rewrite the left-hand side, and then apply Jensen's inequality on the convex function $x^{n-1}, (x>0)$, as follows:
\begin{align*}
&p\left(1-\frac{b}{b+1}p\right)^{n-1}+(1-p)(1-p)^{n-1}\\
&\geq{}\left[p\left(1-\frac{b}{b+1}p\right)+(1-p)(1-p)\right]^{n-1}\\
&=\left[(1-p)+\frac{p^2}{b+1}\right]^{n-1}\\
&\geq\left[(1-p)+\frac{p^{b+1}}{b+1}\right]^{n-1}.
\end{align*}
The last inequality comes from the facts that $b>1$ and $0\leq{}p\leq{}1$.

\end{proof}

This Braess-like paradox is illustrated in Fig.~\ref{Capture5_SNRinf} where $b=5$ and $N_0/P_T \rightarrow 0$. As the figure shows, with the same average transmission probability, the throughput of homogeneous nodes with perfect CSI is never larger than that of homogeneous nodes with no CSI.
%Given that the throughput is a unimodal function of the average transmission probability in the case when CSI is not available (see Appendix \ref{proof_Optimum operating region}), this theorem implies that in order to achieve the same average throughput demand
This means that in order to achieve the same average throughput demand
$\rho$, homogeneous nodes with perfect CSI need to have the average transmission probability $p$ (the smaller solution) at the Nash equilibrium point not smaller than that (the smaller solution) with no CSI when $b>1$.
%In other words, even though the availability of perfect CSI can increase the capacity and should never make the optimal performance (in terms of the average transmission probability or power investment required to sustain the average throughput demand) worse, and it also makes the action space larger which allows for more flexible transmission strategies, the resultant Nash equilibrium due to the best response (selfish) strategy has a worse performance compared to that of the case without CSI (which also uses the best response strategy). We call this phenomenon a Braess-like paradox.

\emph{Discussion}:
\begin{enumerate}
\item
This Braess-like paradox is clearly due to the fact that, while improving a node's received power by transmitting when the channel is better, the threshold strategy also increases the average interference seen by each node. As a result, the SINRs of the nodes are not necessarily higher. Thus, if a node could refrain from taking the best response strategy and be less selfish (e.g., by ignoring the CSI and transmitting with the same probability at all time slots), the other nodes would benefit. If all the nodes could do the same, every node would benefit and the performance would improve. However, for a particular node, doing so would be against its best interest given that it does not know the channels of the other nodes.
%(see the proof of \emph{Proposition \ref{best-response}}).
In addition, an individual node would never be sure if the other nodes would also be altruistic, unless a centralized regularization is applied. Thus, centralized control and/or altruism (or cooperation) are necessary to improve the performance. How the optimal performance can be achieved by centralized control or cooperation is an interesting topic that needs further investigation, but is beyond the scope of this paper.
\item
In general (when $N_0/P_T$ is not sufficiently small), \emph{Theorem \ref{Braess_SINR}} cannot be applied. The throughput comparison given the same average transmission probability depends on the average transmission probability $p$, the total number of nodes $n$, and the capture ratio $b$. Fig.~\ref{Capture5_SNR50} shows the cases with $b=5$ and $P_T/N_0 = 50$. It can be seen that when $n=10$ nodes, the throughput with perfect CSI is slightly higher than the throughput without CSI when the average transmission probability is smaller than 0.05. %Under these conditions, there is no Braess-like paradox.
%For $b<1$, the throughput comparison also depends on the average transmission probability, the total number of nodes and the capture ratio.
%The analysis is more complicated, and we omit it.
Fig.~\ref{Capture0.8_SNR10} is an example with $b=0.8$ and $P_T/N_0 = 10$. It is shown that there is no Braess-like paradox.
%However, the feasible throughput region of the case without CSI is still larger than that of the case with perfect CSI when the number of nodes is 10.
\item
When $N_0 = 0$ and $b = \infty$, we have the collision model. In this case, the throughputs of homogeneous nodes with perfect CSI and without CSI will both be $p(1-p)^{n-1}$ when the average transmission probability is $p$.
\item
For the case with heterogeneous nodes, the conditions for the occurrence of Braess-like paradoxes are quite complicated and need to be analyzed case by case. Here, we present a numerical example of a two-node network. Assume $b=5$ and $\frac{N_0}{P_T}=0.01$.
\begin{itemize}
\item Case (i) $p_1=0.52$ and $p_2=0.24$: the throughput demands achievable for node 1 and node 2 without CSI availability are $\rho_1=0.3957$, $\rho_2=0.129$, respectively. With perfect CSI, the throughput demands achievable for node 1 and node 2 are $\rho_1=0.3952$, $\rho_2=0.118$, respectively. A Braess-like paradox occurs.
\item Case (ii) $p_1=0.580$ and $p_2=0.088$: the throughput demands achievable for node 1 and node 2 without CSI availability are $\rho_1=0.511$, $\rho_2=0.04325$, respectively. With perfect CSI, the throughput demands achievable for node 1 and node 2 are $\rho_1=0.529$, $\rho_2=0.04300$, respectively. The results show that it is possible that node 2 gains while node 1 suffers when CSI is available.
%-----------------------------------------------------------------------------------------
%$\frac{N_0}{P_T}=0.03$.
%    \item Case (i) $p_1=0.52$ and $p_2=0.51$: the throughputs for node 1 and node 2 without CSI availability are $\rho_1=0.25735$, $\rho_2=0.2487$, respectively. With perfect CSI, the throughputs for node 1 and node 2 are $\rho_1=0.25732$, $\rho_2=0.2476$, respectively. A Braess-like paradox occurs.
%\item Case (ii) $p_1=0.221$ and $p_2=0.6$: the throughputs for node 1 and node 2 without CSI availability are $\rho_1=0.09511$, $\rho_2=0.4213$, respectively. With perfect CSI, the throughputs for node 1 and node 2 are $\rho_1=0.09509$, $\rho_2=0.4674$, respectively. The results show that it is possible that node 2 gains while node 1 suffers when CSI is available.
\end{itemize}
\end{enumerate}

\section{Power Capture Model}\label{Section_Power_Capture_Model}
\subsection{Equilibrium point analysis for power capture}\label{Subsection_Equilibrium_Point_power}
We first analyze the case when CSI is not available.
Let the PDF and the cumulative density function (CDF) of $|h_{i,k}|^2$ be $f(x)$ and $F(x)$, respectively.
The received power $P_i$ of node $i$ at time $k$ is given by $B_{i,k}P_T|h_{i,k}|^2$ where $B_{i,k}=1$ if node $i$ transmits at time $k$ and $B_{i,k}=0$ otherwise. Then the average throughput of node $i$ is given by
\begin{align*}
r_i(p_i,\textbf{p}_{-i})
&=Pr\left[B_{i,k}|h_{i,k}|^2>\max_{j\neq i}\left\{(1+\Delta)B_{j,k}|h_{j,k}|^2\right\}\right]\\
&=p_i\int_0^{\infty}\prod_{j\neq{}i}
\left[\int_0^{\frac{{x_i}}{1+\Delta}}\!\!f(x_j)dx_j\!+\!(1\!-\!p_j)\int_{{\frac{{x_i}}{1+\Delta}}}^{\infty}\!\!f(x_j)dx_j\right]\!f({x_i})d{x_i}\notag\\
&=p_i\int_0^{\infty}\prod_{j\neq{}i}
\left[1-p_j\left(1-F\left({\frac{{x_i}}{1+\Delta}}\right)\right)\right]\!f({{x_i}})d{x_i}\notag\\
&=p_i\left\{1-\left(\sum_{j\neq{}i}p_j\right)\int_0^{\infty}\left(1-F\left(\frac{{x_i}}{1+\Delta}\right)\right)f({x_i})d{x_i} \right.\notag\\
&+\left(\sum_{(j,k)\in{}I_{-i}}p_jp_k\right)\int_0^{\infty}\left(1-F\left(\frac{{x_i}}{1+\Delta}\right)\right)^2f({x_i})d{x_i}\notag\\
&-\cdots\notag\\
&\left.+(-1)^{n-1}\left(\prod_{j\neq{}i}p_j\right)\int_0^{\infty}\left(1-F\left(\frac{{x_i}}{1+\Delta}\right)\right)^{n-1}f({x_i})d{x_i}\right\}\notag\\
\end{align*}
where the last equality is obtained by expanding the product terms in the integral, and the notation $(y_1,y_2,\ldots,y_k)\in{}I_{-i}$ means that $y_1<y_2<\cdots<y_k$, all belonging to the node index set $I_{-i}\triangleq\{1,2,\ldots,n\}\setminus\{i\}$.

For Rayleigh fading channels, $f(x)$ is the exponential function, so we have
\[
\int_0^{\infty}\left(1-F\left(\frac{{x_i}}{1+\Delta}\right)\right)^{n}f({x_i})d{x_i}=\frac{1+\Delta}{1+\Delta+n}.
\]

Therefore, if $(p_1,\ldots,p_n)$ is a Nash equilibrium point, we have the following set of equations for the corresponding achievable average throughput demands
\begin{align}\label{eq_NE_rhoi}
r_i(p_i,\textbf{p}_{-i})=\rho_i=p_i&-\frac{1+\Delta}{2+\Delta}\sum_{j\in{}I_{-i}}p_ip_j+\frac{1+\Delta}{3+\Delta}\!\sum_{(j,k)\in{}I_{-i}}\!\!p_ip_jp_k\notag\\
&-\cdots+(-1)^{n-1}\frac{1+\Delta}{n+\Delta}\prod^n_{j=1}p_j,~~\forall{}i=1,\ldots,n,
\end{align}
and the resultant total achievable average throughput demand
\begin{align}
\sum^n_{i=1}\rho_i&\!=\!\sum^n_{i=1}p_i\!-\!\frac{2+2\Delta}{2+\Delta}\sum_{(i,j)\in{}I}p_ip_j\!+\!\frac{3+3\Delta}{3+\Delta}\sum_{(i,j,k)\in{}I}p_ip_jp_k\notag\\
&\hspace{0.45in}-\cdots+(-1)^{n-1}\frac{n+n\Delta}{n+\Delta}\prod^n_{i=1}p_i,
\end{align}
where the notation $(y_1,y_2,\ldots,y_k)\in{}I$ means that $y_1<y_2<\cdots<y_k$, all belonging to the node index set $I\triangleq\{1,2,\ldots,n\}$.

\begin{example}[$\Delta=\infty$]
The set of equations for a Nash equilibrium point become $\rho_i=p_i\prod_{j\neq i}(1-p_j)$, i.e., the case $\Delta=\infty$ corresponds to the collision model, and there exist exactly two Nash equilibrium points for any throughput demands within the feasible region \cite{Aloha_JSAC08}.
\end{example}

\begin{example}[$\Delta=0$]
For the special case where $\Delta=0$, i.e., perfect power capture model, we have
\begin{align}\label{Throughput_i_power}
\rho_i=p_i&-\frac{1}{2}\sum_{j\in{}I_{-i}}p_ip_j+\frac{1}{3}\!\sum_{(j,k)\in{}I_{-i}}\!\!p_ip_jp_k\notag\\
&-\cdots+\frac{(-1)^{n-1}}{n}\prod^n_{j=1}p_j,
\end{align}
and the total achievable average throughput demand
\begin{align}\label{Sum_rho}
\sum^n_{i=1}\rho_i&=\sum^n_{i=1}p_i-\sum_{(i,j)\in{}I}p_ip_j+\sum_{(i,j,k)\in{}I}p_ip_jp_k \notag\\
&\ \ \ -\cdots+(-1)^{n-1}\prod^n_{i=1}p_i\notag\\
&=1-\prod^n_{i=1}(1-p_i).
\end{align}

The following theorem shows that there exists a \emph{unique} Nash equilibrium point:
\begin{theorem}\label{NE for power capture}
Under the perfect power capture model ($\Delta=0$) with no CSI available to all nodes, there exists a \emph{unique} Nash equilibrium point for the average throughput demands $(\rho_1,\ldots,\rho_n), \rho_i\geq{}0, \forall{}i$,
\emph{if and only if}
\[
\sum^n_{i=1}\rho_i\leq{}1
\]
\end{theorem}
\begin{proof}
In Appendix \ref{proof_NE_for_power_capture}.

\end{proof}

%When $\Delta=\infty$ (the collision model), there exist exactly two Nash equilibrium points for any throughput demands within the feasible region \cite{Aloha_JSAC08}. When $\Delta=0$ (the perfect power capture model), there is a unique Nash equilibrium point.
%From these two extreme cases, we infer that for a general $\Delta$, there exist at most two Nash equilibrium points, since the feasible region shrinks as $\Delta$ increases due to the capture model (\ref{eq_power_capture_model}).
%However, analysis of the Nash equilibrium for a general $\Delta$ is quite complicated with no simple exact form.
\end{example}

In the case when perfect CSI is available to all nodes, a particular node $i$ will transmit only when its CSI is larger than a threshold $T_i$. $T_i$ must satisfy $\int^\infty_{T_i} f(x_i)dx_i=e^{-T_i}=p_i$, where $p_i$ is the average transmission probability of node $i$, and $f(\cdot)$ is the exponential PDF for Rayleigh fading channels. The average throughput of node $i$ can be computed as
\begin{align}
r_i(p_i,\textbf{p}_{-i})&=\int_{T_i}^{\infty}\prod_{j\neq{}i}
\left[\int_0^{T_j}\!\!f(x_j)dx_j\!+\!\int_{T_j}^{\max\left\{\frac{x_i}{1+\Delta}, ~T_j\right\}}\!\!f(x_j)dx_j\right]\!f(x_i)dx_i \notag \\
&=\int_{T_i}^{\infty}\prod_{j\neq{}i}
\left[\max \left\{1-p_j, ~1-e^{-\frac{x_i}{1+\Delta}}\right\}\right]\!f(x_i)dx_i.\label{thrpt_eq_CSI_power_capture}  \end{align}
This expression depends on the specific values of the thresholds (hence the average transmission probabilities and the average throughput demands) of individual nodes.

In the following, we will only consider the case with homogeneous nodes. Let the throughput demands be $(\rho,\ldots,\rho)$ and the Nash equilibrium point $(p_1,\ldots,p_n)=(p,\ldots,p)$. When CSI is not available to all nodes, from (\ref{eq_NE_rhoi}) we have
\begin{align}\label{eq_rho_homogeneous}
r_i(p,\ldots, p)=\rho=p\!-\!{{n-1}\choose{1}}\frac{1+\Delta}{2+\Delta}p^2
   \!+\!\cdots\!+\!(-1)^{n-1}{{n-1}\choose{n-1}}\frac{1+\Delta}{n+\Delta}p^n.
\end{align}
%where ${{n}\choose{k}}=\frac{n!}{(n-k)!k!}$ is the number of $k$-combinations from $n$ elements.

In the case when perfect CSI is available to all nodes, a node will transmit only when its CSI is larger than a threshold $T$. $T$ must satisfy $\int^\infty_T f(x)dx=e^{-T}=p$, where $p$ is the average transmission probability.
%To have the same average transmission probability $p$ (or average power investment) as that of the case without CSI, we need $\int^\infty_T f(x)dx=p$, where $f(\cdot)$ is the exponential PDF for Rayleigh fading channels.
For this case, denote the average throughput demand achievable at the Nash equilibrium point $(p,\ldots,p)$ as $\rho'$. We have
\begin{align}\label{eq_rho'_homogeneous}
r_i(p,\ldots, p)=\rho'=&\int^{(1+\Delta)T}_T\left[ \prod_{j\neq i} \int^T_0f(x_j)dx_j \right] f({x_i})d{x_i}\notag\\
&+\int^{\infty}_{(1+\Delta)T}\left[\prod_{j\neq i}\int^{\frac{{x_i}}{1+\Delta}}_0f(x_j)dx_j \right] f({x_i})d{x_i}\notag\\
=&\left[p(1-p)^{n-1}(1-p^{\Delta})\right]\!+\!\Big[p^{1+\Delta}-{{n-1}\choose{1}}\frac{1+\Delta}{2+\Delta}p^{2+\Delta}\notag\\
   &+\cdots+(-1)^{n-1}{{n-1}\choose{n-1}}\frac{1+\Delta}{n+\Delta}p^{n+\Delta}\Big].
\end{align}

The Nash equilibrium can be analyzed as follows when perfect CSI is available to homogeneous nodes.
From (\ref{eq_rho'_homogeneous}), we have $\frac{dr_i}{dp}=(1-p)^{n-2}\left[1-np+(n-1)p^{1+\Delta}\right]$, and
\begin{align}\label{eq_CSI_Delta}
\frac{d}{dp}\left(1-np+(n-1)p^{1+\Delta}\right)=-n+(n-1)(1+\Delta)p^{\Delta}.
\end{align}
Note that $1-np+(n-1)p^{1+\Delta}=1$ if $p=0$, and $1-np+(n-1)p^{1+\Delta}=0$ if $p=1$. We consider two cases: $0\le\Delta\le\frac{1}{n-1}$ and $\Delta>\frac{1}{n-1}$ separately in the following.
\begin{itemize}
\item
Case (i) $0\le\Delta\le\frac{1}{n-1}$: From (\ref{eq_CSI_Delta}), $1-np+(n-1)p^{1+\Delta}$ is a decreasing function for $p\in[0,1]$, so $\frac{dr_i}{dp}\ge 0$ when $0\le\Delta\le\frac{1}{n-1}$. This means that $r_i$ is an increasing function in $p$, so there exists a \emph{unique} Nash equilibrium point if the throughput demand is achievable under the case $0\le\Delta\le\frac{1}{n-1}$.
\item
Case (ii) $\Delta>\frac{1}{n-1}$: $1-np+(n-1)p^{1+\Delta}$ has a minimum at $p^*=\left[\frac{n}{(n-1)(1+\Delta)}\right]^{\frac{1}{\Delta}}\in(0,1)$, and $1-np+(n-1)p^{1+\Delta}$ is decreasing in $(0,p^*)$ and increasing in $(p^*,1)$. Therefore, $1-np+(n-1)p^{1+\Delta}$ has exactly two zeros $z_1$, $z_2(=1)$ in $[0,1]$. It follows that $\frac{dr_i}{dp}\ge 0$ in $[0,z_1]$ and $\frac{dr_i}{dp}\le 0$ on $[z_1,1]$. In other words, $r_i$ increases to the maximum $\rho'_{max}$ when the transmission probability $p$ is from $p=0$ to $p=z_1$ and then decreases when $p>z_1$. This means that there exist at most two Nash equilibrium points if the throughput demand is achievable (i.e., $\rho'\le\rho'_{max}$) under the case $\Delta>\frac{1}{n-1}$.
\end{itemize}

\subsection{Braess-like paradox for power capture}\label{Subsection_Paradox_powercapture}
We present analytically a Braess-like paradox for the case with homogeneous nodes by the following theorem.
%Comparing (\ref{eq_rho_homogeneous}) and (\ref{eq_rho'_homogeneous}), we have the following theorem.
\begin{theorem}
With the same average transmission probability $p$ (or average power investment), the achievable average throughput demand of homogeneous nodes with perfect CSI is not larger than that of homogeneous nodes without CSI under the power capture model.
\end{theorem}
\begin{proof}
We want to show $\rho'\le\rho$, where $\rho'$ and $\rho$ are given in (\ref{eq_rho_homogeneous}) and (\ref{eq_rho'_homogeneous}), respectively. We have
\begin{align*}
\rho'=&(1-p^{\Delta})p(1-p)^{n-1}+p^{\Delta}\rho\\
     \le&\max\{p(1-p)^{n-1},\rho \}\\
     \le&\rho,
\end{align*}
where the first inequality is deduced from the convex combination of $p(1-p)^{n-1}$ and $\rho$, and the second inequality is due to the fact that $p(1-p)^{n-1}$ equals to the throughput for $\Delta=\infty$ (when CSI is not available) and the throughput decreases as $\Delta$ increases. The proof is complete.
\end{proof}

Similarly, this theorem implies a Braess-like paradox that for the same achievable average throughput demand $\rho$, homogeneous nodes with perfect CSI have the average transmission probability $p$ (or average power investment) at the Nash equilibrium point not smaller than that of the case without CSI.

\section{Distributed Algorithms}
\label{Section_Distributed_Algorithm}

In a network, each node can usually estimate its average throughput through, for example, the acknowledgement of successful packet reception from the BS. Let $\hat{\rho}_i$ denote node $i$'s throughput estimate. We can set node $i$'s initial (i.e., at the $0$th iteration) transmission probability $p_i(0)=\rho_i$, as the transmission probability needs to be at least $\rho_i$ to fulfill the throughput demand. We provide one most common distributed mechanism converging to the Nash equilibrium point. The readers are referred to \cite{Distributed_Update_JSAC08}\cite{Sabir09}\cite{Aloha_JSAC08} for more discussions on the use of various distributed algorithms to achieve the equilibrium points in random access games.
At the ($m+1$)th iteration, each node updates its transmission probability by
\begin{align*}
p_i(m+1)=p_i(m)+\epsilon(m)\left[\min\left(1,\frac{\rho_i}{\hat{\rho}_i(m)}p_i(m)\right)-p_i(m)\right],
\end{align*}
where the step size $ \epsilon(m) > 0$. %When $\epsilon(m)=1$, we recover the best response mechanism.
Usually, $\epsilon(m)\le 1$, for example, $\epsilon(m)=\frac{1}{1+m}$ and a smaller $\epsilon(m)$ will more likely ensure the convergence to the better Nash equilibrium point (i.e., $\sum p_i\le \frac{b+1}{b}$) under the SINR capture model, and to the unique Nash equilibrium under the power capture model. However, it also takes a longer time for convergence.
When CSI is available, the threshold $T_i(m+1)$ can be uniquely determined from $p_i(m+1)$ and vice versa if the channel characteristics are known.
Note that the time between two consecutive iterations of a node does not have to be the same as that of the other nodes. The nodes can update their transmission probabilities synchronously or asynchronously.
%The best response is obtained by observing that given other nodes' transmission probabilities, the throughput of node $i$ is proportional to its transmission probability in both the power capture and the SINR capture models. Therefore, it updates its transmission probability by the factor $\rho_i/\hat{\rho}_i(m)$ at the ($m+1$)th iteration.
%Due to the throughput estimation error, the updated transmission probability could be large than 1. In this case the node simply transmits with probability $1$.

Issues such as infeasible throughput demands, and the design tradeoff of $\epsilon(m)$ between ensuring convergence and the convergence time, are beyond the scope of this paper.

We present some simulation results in Fig.~\ref{sim_noCSI} and Fig.~\ref{sim_perfectCSI}, where the channels are i.i.d. Rayleigh fading, and the reception model is SINR capture.
We consider a three-node network with throughput demands $\rho_1=0.10$, $\rho_2=0.05$, $\rho_3=0.01$, the capture ratio $b=5$, and $\frac{P_T}{N_0}=10$.
The cases with no CSI and perfect CSI are considered, and node $i$ estimates its $\hat{\rho}_i$ by the number of successfully transmitted packets divided by number of time slots that have elapsed.
For the case with no CSI, the dynamic of the transmission probabilities is plotted in Fig.~\ref{sim_noCSI}.
For the case with perfect CSI, we have the threshold $T_i(m+1)=-\ln p_i(m+1)$, and node $i$ will transmit only if its channel gain is larger than the threshold $T_i$. The dynamic of the thresholds is plotted in Fig.~\ref{sim_perfectCSI}.
Three realizations are shown for each case, and the figures illustrate that they indeed converge to the same equilibrium.
%and update its transmission probability (or threshold) each time a packet is successfully transmitted.
%or update its transmission probability (or threshold) every time slot
%This is a asynchronously update mechanism.

\section{Conclusion}\label{Section_Conclusion}
In this paper, we used a game-theoretic approach to study the Nash equilibrium point of CSI-dependent transmission probabilities for selfish random access nodes in fading channels with capture.
The analysis revealed that under the power capture and SINR capture models, there are at most two Nash equilibrium points in the feasible region of throughput demands. For the collision channel, which is a special case of both the power capture and SINR capture models, there are exactly two Nash equilibrium points within the feasible region. On the other hand, there is one unique Nash equilibrium point under the perfect power capture model when CSI is not available to selfish nodes.
Our work extends the existing works in the literature. Moreover,
we pointed out that, in some situations, performance degradation may occur when CSI is available to selfish random access nodes as compared to when CSI is not available. We called this phenomenon a Braess-like paradox. In particular, we analytically showed that for homogeneous nodes, Braess-like paradoxes occur in the power capture model and
in the SINR capture model with the capture ratio larger than one and noise to signal ratio sufficiently small.

\section*{Acknowledgment}
The authors would like to thank the anonymous reviewers and the Associate Editor for their valuable comments on how to improve the paper.
The authors would also like to thank Professor Fedja Nazarov (University of Wisconsin, Madison) for the help and discussion on the proof of \emph{Theorem \ref{NE for power capture}}.

\appendix

\subsection{\textbf{Proof of Theorem \ref{NE for power capture}}}
\label{proof_NE_for_power_capture}
We need to show that for any throughput demands $(\rho_1,\ldots,\rho_n)$ satisfying $\rho_i\geq{}0$ and $\sum^n_{i=1}\rho_i\leq{}1$, there is a \emph{unique} solution $(p_1,\ldots,p_n)$ called the Nash equilibrium point.
That is, the representation (\ref{Throughput_i_power}) is a \emph{one-to-one and onto} mapping from the unit $n$-cube ($0\leq{}p_i\leq{}1$) to the unit $n$-simplex ($\rho_i\geq{}0$ and $\sum^n_{i=1}\rho_i\leq{}1$).

In the following, we simply use $\sum_i$ to denote $\sum_{i=1}^n$ and a similar expression for $\prod_j$.
Consider the auxiliary function $G$
\begin{align}
G(p_1,\ldots,p_n)=\sum_i \rho_i\ln p_i+\int_0^1\frac{\left[\prod_j(1-p_jx)\right]-1}{x}\,dx.
\end{align}
It can be checked that the solution of the set of equations (\ref{Throughput_i_power}) is a critical point of this function. We first show that the representation (\ref{Throughput_i_power}) maps the open cube ($0<p_i<1$) one-to-one and onto the open simplex ($\rho_i>0$ and $\sum^n_{i=1}\rho_i<1$), and then deal with the boundary.

Let $p_i=e^{t_i}$. We will show that $G$ is a strictly concave function of variables $t_i$ as long as all $t_i<0$ (correspondingly, $0<p_i<1$).
%Indeed, its second differential
%\begin{align*}
%D^2F(h)=\sum_{i,j}\frac{\partial^2F}{\partial t_i\partial t_j}h_ih_j
%\end{align*}
This is equivalent to showing that the Hessian of $G$ is positive semi-definite (or nonnegative definite), that is,
$\sum_{i,j}\frac{\partial^2G}{\partial t_i\partial t_j}h_ih_j< 0$.
Let $\Phi(x) = \prod_j(1 - p_j x)$ and $\psi_i(x) = \frac {p_i}{1 - p_ix}$, we have the following equation:
\begin{align}
&\sum_{i,j}\frac{\partial^2G}{\partial t_i\partial t_j}h_ih_j\notag\\
=&\sum_{i\neq j}\frac{\partial^2G}{\partial t_i\partial t_j}h_ih_j+\sum_i\frac{\partial^2G}{\partial t^2_i}h^2_i\notag\\
=&\int_0^1 \left[\left(\sum_i h_i\psi_i(x)\right)^2-\sum_ih^2_i\psi_i^2(x) \right]x\Phi(x) \,dx \notag\\
&+ \sum_i\int_0^1 \left(-h_i^2 \psi_i(x)\Phi(x)\right)\,dx \notag\\
=&\int_0^1 \left[\sum_i h_i\psi_i(x) \right]^2x\Phi(x) \,dx \notag\\
&- \int_0^1\sum_i h_i^2 (\psi_i(x) + x\psi_i^2(x))\Phi(x)\,dx.
\end{align}
By Cauchy-Schwarz inequality, we have
\begin{align}\label{Cauchy-Schwarz}
\left[\sum_i h_i\psi_i(x) \right]^2\le \left[\sum_i h_i^2\psi_i(x) \right]\cdot \left[\sum_i \psi_i(x) \right].
\end{align}
Hence, all we need is to show that
\begin{align*}
\sum_jh^2_j \int_0^1 \psi_j(x) \left[\sum_i \psi_i(x) \right]x\Phi(x) \,dx \\
- \sum_j h_j^2\int_0^1 \left( \psi_j(x) + x\psi_j^2(x) \right) \Phi\,dx < 0,
\end{align*}
which is true if for all $j$ we have
\begin{align*}
\int_0^1\! \psi_j(x)\! \left[\sum_i \psi_i(x) \right]\!x\Phi(x) \,dx\!\le\!
\int_0^1\! \left( \psi_j(x)\!+\!x\psi_j^2(x)\right) \Phi(x)\,dx,
\end{align*}
or, equivalently,
\begin{align}\label{eq_thm4_final}
\int_0^1 \psi_j(x) \left[\sum_{i\ne j} \psi_i(x) \right]x\Phi(x) \,dx\le \int_0^1 \psi_j(x) \Phi(x)\,dx.
\end{align}
Let $\Phi_j(x)=\prod_{i\ne j}(1-p_ix)$ and $\Phi_j'(x)=\frac{d}{dx}\Phi_j(x)$. The inequality (\ref{eq_thm4_final}) can be rewritten as
\begin{align*}
&p_j\int_0^1 -x\Phi_j'(x)\,dx\le p_j\int_0^1\Phi_j(x)\,dx \\
\Leftrightarrow~
&p_j\int_0^1 x\Phi_j'(x)+\Phi_j(x)\,dx=p_j\left[x\Phi_j(x)\right]^1_{0}=p_j\Phi_j(1)> 0,
\end{align*}
which is obvious.

Hence, we know that the function $G$ is strictly concave in $t_i$ when all $t_i < 0$.
It follows that there is at most one critical point in the open unit cube and that point is the point of maximum.
This takes care of the one-to-one part.

As for the onto part, we only need to show that for $(\rho_1,\ldots,\rho_n)$ satisfying $\rho_i>0$ and $\sum^n_{i=1}\rho_i<1$, the maximum is attained inside the cube, not on the boundary (i.e., some $p_i=1$).

Note that the maximum cannot be attained with any $p_i=0$ (then $G=-\infty$).
Hence, $\frac{\partial G}{\partial p_i}\geq 0$ at the point of maximum whether it is on the boundary or not (otherwise slightly shifting the point to the left will result in a bigger value). Thus,
\begin{align*}
&\sum_i \frac{\partial G}{\partial p_i}\geq 0 \\
\Rightarrow&\sum^n_{i=1}\rho_i\geq 1-\prod^n_{i=1}(1-p_i),
\end{align*}
where we have use the equality (\ref{Sum_rho}). Since the left hand side is smaller than $1$, we cannot have any $p_i=1$ at the point of maximum, and thereby we have proved that the open cube is one-to-one and onto mapped to the open simplex.

Now, we look at the boundary issue. The representation (\ref{Throughput_i_power}) is a continuous function, and the unit cube is compact, so the image has to be compact. The onto claim is done.

It is clear that if $\rho_i=0$, we must have $p_i=0$. So, removing all zeroes, we can reduce the problem to itself with fewer variables. That is, we only need to consider all $\rho_i>0$. Note that in (\ref{Cauchy-Schwarz}) we have the equality only when all $h_i$ are the same (the only direction in which we may lack strict concavity). Hence, if there are two critical points, both points cannot be on the boundary simultaneously, and one has to be inside the cube. Once one of them is inside the cube, we immediately get $\sum^n_{i=1}\rho_i < 1$, no critical point on the boundary at all then. In conclusion, we have at most one critical point. The one-to-one claim is done, and the proof is complete.

%\bibliographystyle{IEEEtran}

\newpage
\begin{figure}[!t]
\centering
\includegraphics[width=0.7\textwidth]{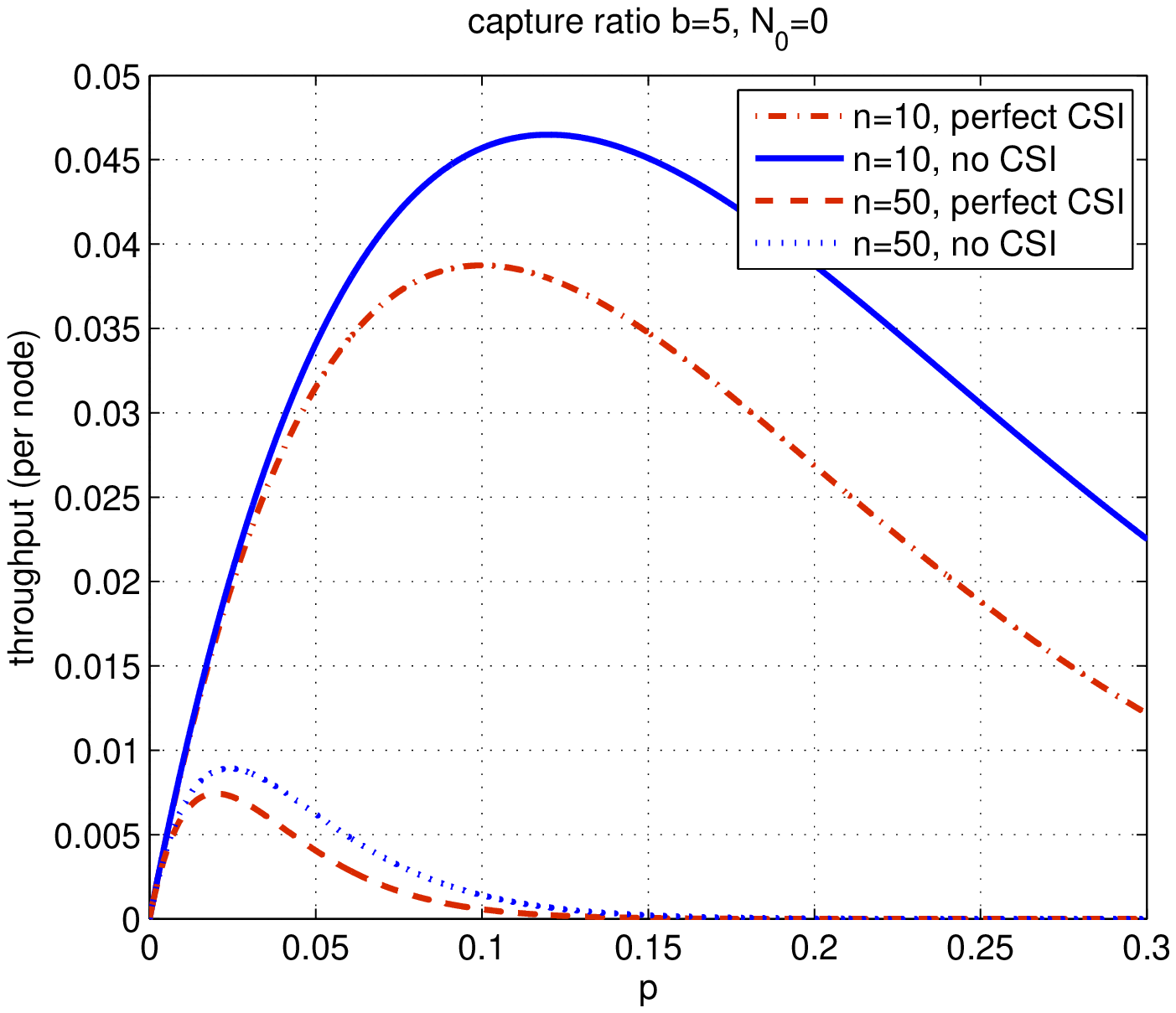}
\caption{$n$ homogeneous nodes with perfect CSI and no CSI at the Nash equilibrium point under the SINR capture model with the capture ratio $b=5$ and $N_0/P_T \rightarrow 0$.}
\label{Capture5_SNRinf}
\end{figure}

\begin{figure}[!t]
\centering
\includegraphics[width=0.7\textwidth]{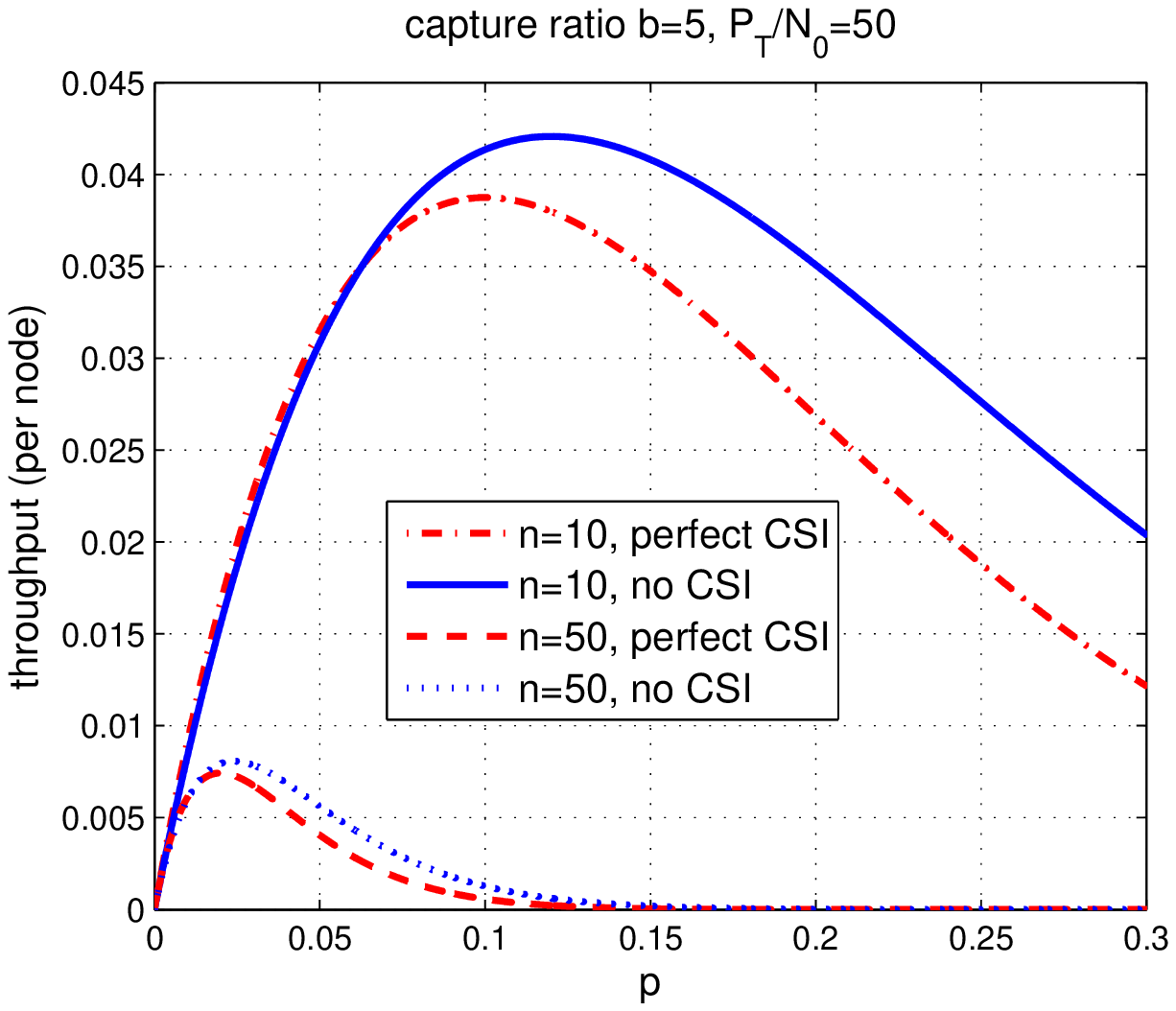}
\caption{$n$ homogeneous nodes with perfect CSI and no CSI at the Nash equilibrium point under the SINR capture model with the capture ratio $b=5$ and $P_T/N_0=50$.}
\label{Capture5_SNR50}
\end{figure}

\begin{figure}[!t]
\centering
\includegraphics[width=0.7\textwidth]{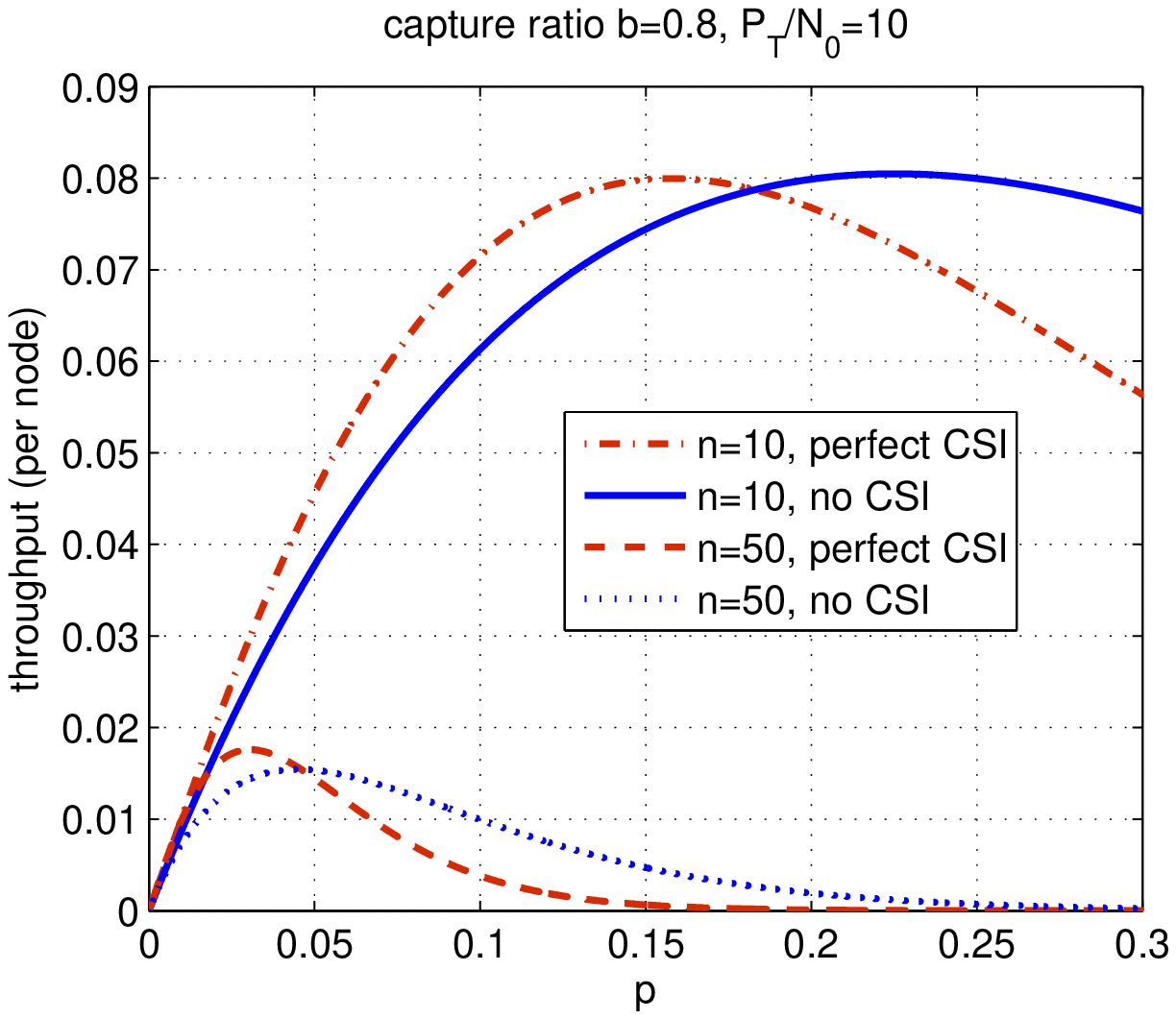}
\caption{$n$ homogeneous nodes with perfect CSI and no CSI at the Nash equilibrium point under the SINR capture model with the capture ratio $b=0.8$ and $P_T/N_0=10$.}
\label{Capture0.8_SNR10}
\end{figure}

\begin{figure}[!t]
\centering
\includegraphics[width=0.75\textwidth]{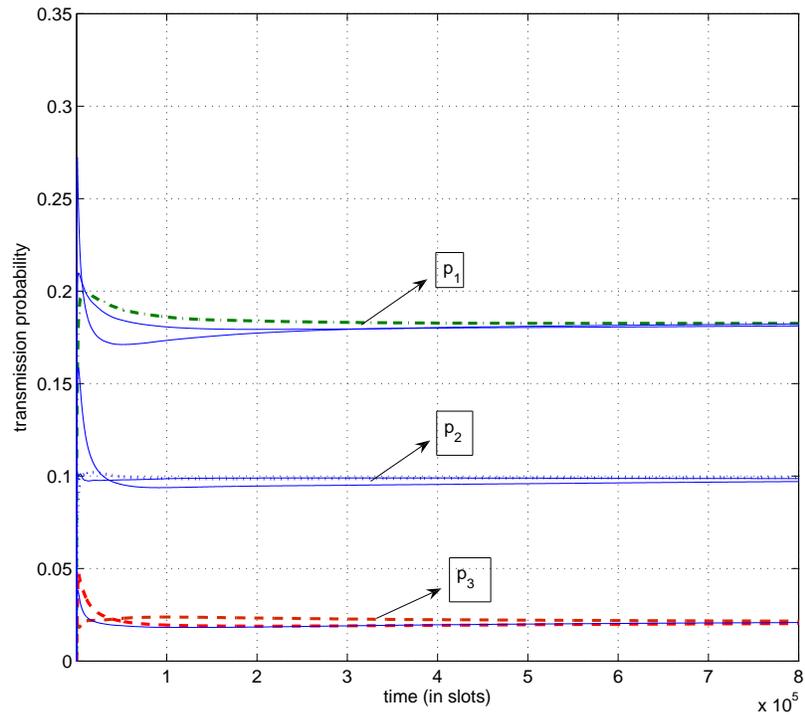}
\caption{Dynamic of transmission probabilities in a three-node network with no CSI under the SINR capture model. The capture ratio $b=5$, $\frac{P_T}{N_0}=10$, and the throughput demands are $\rho_1=0.10$, $\rho_2=0.05$, $\rho_3=0.01$. }
\label{sim_noCSI}
\end{figure}

\begin{figure}[!t]
\centering
\includegraphics[width=0.75\textwidth]{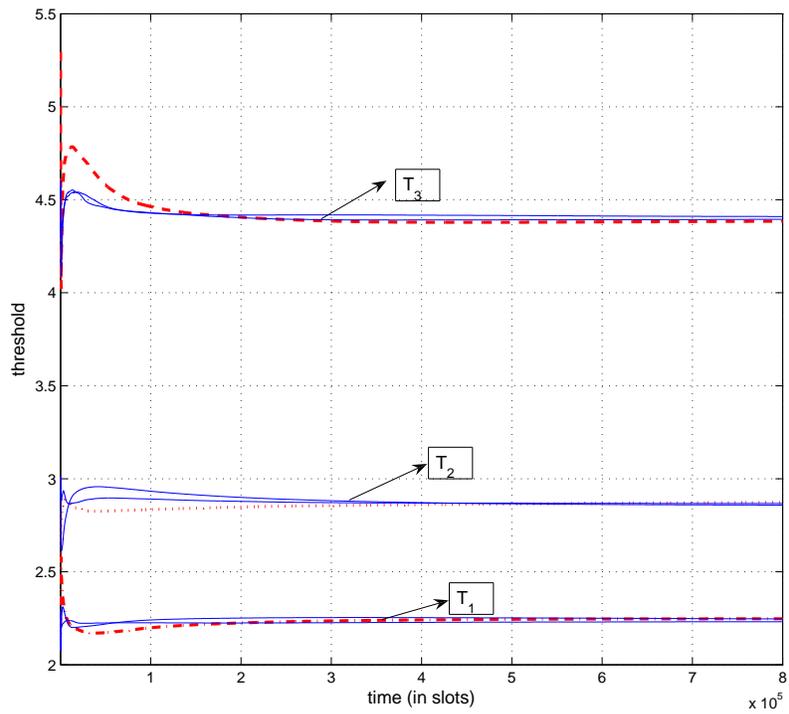}
\caption{Dynamic of thresholds in a three-node network with perfect CSI under the SINR capture model. The capture ratio $b=5$, $\frac{P_T}{N_0}=10$, and the throughput demands are $\rho_1=0.10$, $\rho_2=0.05$, $\rho_3=0.01$. }
\label{sim_perfectCSI}
\end{figure}

\end{document}